\documentclass[draft]{svjour3}

\usepackage[utf8]{inputenc} 
\usepackage[T1]{fontenc}    %
\usepackage[english]{babel} 
\usepackage{graphicx} 
\usepackage{xcolor}
\usepackage{hyperref} 
\usepackage{comment}
\usepackage{amsmath,amsfonts,amssymb,bm} 

\newcommand{\dv}{\mathrm d}
\newcommand{\tr}{\mathrm {Tr}}
\newcommand{\diag}{\mathrm {diag}}
\newcommand{\E}{\mathbb {E}}
\newcommand{\U}{\mathrm {U}}
\newcommand{\s}{\mathcal {S}}

\newcommand{\Z}{\mathbb {Z}}
\newcommand{\C}{\mathbb {C}}
\newcommand{\eins}{\leavevmode\hbox{\small1\kern-3.8pt\normalsize1}}

\begin{document}

\title{Harmonic analysis for rank-$1$ Randomised Horn Problems}

\author{ Jiyuan Zhang \and Mario Kieburg \and Peter J. Forrester}
\institute{J. Zhang \and M. Kieburg \and P. J. Forrester \at School of Mathematics and Statistics, The University of Melbourne, Victoria 3010, Australia \\  \email{jiyuanz@student.unimelb.edu.au, m.kieburg@unimelb.edu.au, pjforr@unimelb.edu.au}}
\date{\today}

\maketitle

\begin{abstract}  
The randomised Horn problem, in both its additive and multiplicative version, has recently drawn increasing interest. Especially, closed analytical results have been found for the rank-$1$ perturbation of sums of Hermitian matrices and products of unitary matrices. We will generalise these results to rank-$1$ perturbations for products of positive definite Hermitian matrices and prove the other results in a new unified way. Our ideas work along harmonic analysis for matrix groups via spherical transforms that have been successfully applied in products of random matrices in the past years. In order to achieve the unified derivation of all three cases, we define the spherical transform on the unitary group and prove its invertibility.

\end{abstract}

\section{Introduction and Main Results}\label{S1}

In 1962, Horn~\cite{Ho62} raised a question on finding the possible range of the eigenvalues $\mathbf c=\diag(c_1,\cdots,c_n)$ of the sum $C=A+B$ of two fixed $n\times n$ Hermitian matrices $A$ and $B$ whose eigenvalues $\mathbf a=\diag(a_1,\cdots,a_n)$ and $\mathbf b=\diag(b_1,\cdots,b_n)$ are given. Unfortunately, only one equation can be exactly found in the general setting, namely the traces of the matrices
\begin{equation}\label{1.01}
\tr\, A +\tr\, B=\sum_{j=1}^n(a_j+b_j)=\sum_{j=1}^nc_j=\tr\, C,
\end{equation} 
while all other relations can be only expressed in inequalities. These inequalities give a bounded domain on the hyperplane defined by the trace condition~\eqref{1.01}. In the case of $B$ being of rank $1$, the inequalities simplify to the Cauchy interlacing condition,
\begin{equation}\label{1.02}
c_1\ge a_1\ge c_2\ge a_2\ge\ldots\geq c_n\ge a_n,
\end{equation}
when imposing the ordering $c_1\ge \ldots\ge c_n$ and $a_1\ge \ldots\ge a_n$.
This set of inequalities for the general case was proved to be necessary and sufficient by Knutson and Tao~\cite{KT99}, using a combinatorial method.

Horn's problem can be encountered in various fields, such as in representation theory~\cite{Klyachko1998,Fulton2000,Coquereaux2017}, combinatorics~\cite{King2006}, algebraic geometry~\cite{Knutson2000}, quantum information~\cite{Klyachko2004,Zhang2002,Zhang2019} and, indeed, linear algebra~\cite{Bhatia2001}.
In recent years a randomised version of Horn's problem has been considered where $\mathbf a$ and $\mathbf b$ are fixed, while the diagonalizing unitary matrices $U,V$, i.e. $A=U\mathbf aU^\dagger$ and $B=V\mathbf bV^\dagger$ with ``$\dagger$'' the Hermitian adjunct, are drawn from the Haar measure $\mu(\dv U)$ of the unitary group $\U(n)$. Horn's question can then be rephrased to an explicit expression for the joint eigenvalue probability density of the eigenvalues $\mathbf c$. There are general discussions on this particular randomised Horn problem~\cite{Zu18} as well as specialisations to a rank $1$ matrix $\mathbf b$, see~\cite{FG06,Fa19,FZ19}. In the latter case, a closed analytic expression of the joint eigenvalue density is accessible; see also~\eqref{1.1}.

In~\cite{Zu18,FZ19}, an harmonic analysis approach via the  Fourier transform has been suggested. It works along the same ideas of characteristic functions in probability theory, where the density of the sum of independent random variables can be obtained by taking the inverse transform of the product of their corresponding characteristic functions. This approach is commonly seen in proofs of the central limit theorem~\cite{GK49}, for example.

Harmonic analysis for matrices, especially groups has been introduced in the 50's and 60's, see, e.g., the textbook by Helgason~\cite{He84}. In this framework the univariate Fourier analysis is generalized to the \textit{spherical transform}. Instead of considering functions defined on the real line, they now live on a coset $G/K$, where $G$ is a semi-simple Lie group and $K$ is its compact subgroup. The coset space is equipped with a group operation of $G$, which allows to construct a convolution theorem with respect to this group action. Recently for sums of matrices this has been applied in~\cite{KR17}. The advantage of this particular kind of harmonic analysis is that it can be extended to multiplicative group actions as well. There are various works developing the spherical transform for product of random matrices~\cite{KK16,KK19,KFI17,Kieburg2019}.

With this analytical tool in mind, one can address a multiplicative version of Horn's problem. One version of such a question is the projection of a Hermitian matrix $A$ to a co-rank $1$ hyperplane, see~\cite{Baryshnikov,Forrester-Rains,FZ19}, where a new and explicit closed expression of the joint eigenvalue density can be derived. We will study a related problem where $A$ and $B$ are either positive definite Hermitian matrices and the product matrix is $C=A^{1/2}BA^{1/2}$ or are unitary with $C=AB$. When $A$ and $B$ are given as before by their fixed eigenvalues $e^{\mathbf{a}}, e^{\mathbf{b}}$ or $e^{i\mathbf{a}}, e^{i\mathbf{b}}$, respectively, and their diagonalizing Haar-distributed unitary matrices $U$ and $V$. The sum on ${\rm Herm}(n)$ is naturally related to its compact group $\U(n)=\exp[i {\rm Herm}(n)]$ and non-compact realization ${\rm Herm}_+(n)=\exp[ {\rm Herm}(n)]$. Therefore, it is not really surprising that all three cases can be treated in the same framework of spherical transforms, though there are subtle differences as we will see. In general, we have as for the sum an exact equality,
\begin{equation}\label{determinant}
\det A \det B=\det C
\end{equation}
replacing the trace condition~\eqref{1.01}, and a system of inequalities. For the randomised co-rank $1$ case, we will prove the following joint eigenvalue densities using spherical transforms.

\begin{theorem}[Randomized Horn's Problem with a Rank-$1$ Matrix]\label{prop}
Let $\chi$ be the indicator function, $\delta(.)$ be the Dirac delta function, $\Delta(\mathbf{a})=\prod_{k>l}(a_k-a_l)$ be the Vandermonde determinant, and choose two real diagonal matrices $\mathbf a=\diag(a_1,\ldots, a_n)$ and $\mathbf b=\diag(b,0,\ldots,0)$ satisfying $a_1>\ldots>a_n$ and $b>0$. In all three cases below, we assume the eigenvalues  $\mathbf c=\diag(c_1,\ldots,c_n)$ of $C$ to be ordered $c_1>\ldots>c_n$, as well. Moreover, we choose a fixed or random $U\in\U(n)$ and a Haar-distributed $V\in\U(n)$.
\begin{enumerate}
	\item (Sum on $\mathrm{Herm}(n)$, see~\cite{FG06,FZ19,Fa19}) The eigenvalues  $\mathbf c$ of $C=U\mathbf aU^\dagger+V\mathbf bV^\dagger$ are distributed according to the joint probability distribution
	\begin{equation}\label{1.1}
		p(\mathbf{c})=\frac{(n-1)!}{b^{n-1}}\frac{\Delta(\mathbf c)}{\Delta(\mathbf a)}\,\delta\left(b+\sum_{j=1}^{n}(a_j-c_j)\right)\chi_{c_1>a_1>\cdots>c_n>a_n},
	\end{equation}
	for almost all $\mathbf{c} \in \mathbb{R}^n$ with $\min_{k,l}|a_k-c_l|>0$.

	\item (Product on $\mathrm{Herm}_+(n)$) The eigenvalues $e^{\mathbf{c}}$ of $C=Ue^{\mathbf a/2}U^\dagger Ve^{\mathbf b}V^\dagger Ue^{\mathbf a/2}U^\dagger$ follow the joint probability distribution
	\begin{equation}\label{1.2}
		p(e^{\mathbf{c}})=\frac{(n-1)!}{(e^b-1)^{n-1}}\frac{\Delta(e^{\mathbf c})}{\Delta(e^{\mathbf a})}\,\delta\left(b+\sum_{j=1}^{n}(a_j-c_j)\right)\chi_{c_1>a_1>\cdots>c_n>a_n>0},
	\end{equation}
	for almost all $\mathbf{c}\in\mathbb{R}^n$ with $\min_{k,l}|a_k-c_l|>0$.
	
	\item (Product on $\U(n)$, see~\cite{FZ19}) Restricting $a_1,\cdots, a_n,b\in(-\pi,\pi]$, the eigenvalues  $e^{i\mathbf{c}}$ of $C=Ue^{i\mathbf a}U^\dagger Ve^{\mathbf b}V^\dagger$ have the joint probability distribution
	\begin{equation}\label{1.3}
		p(e^{i\mathbf{c}})=\frac{i^{n-1}(n-1)!}{(e^{ib}-1)^{n-1}}\frac{\Delta(e^{i\mathbf c})}{\Delta(e^{i\mathbf a})}\,\delta_{2\pi}\left(b+\sum_{j=1}^{n}(a_j-c_j)\right)\chi_{2\pi+a_n\ge c_1>a_1>\cdots>c_n>a_n},
	\end{equation}
	for almost all $\mathbf{c}\in(a_n,2\pi+a_n)^n$ with $\min_{k,l}{\rm mod}_{2\pi}(a_k-c_l)>0$ and $\max_{k,l}$ ${\rm mod}_{2\pi}(a_k-c_l)<2\pi$. Here, $\delta_{2\pi}$ is the $2\pi$-periodic Dirac delta function and ${\rm mod}_{2\pi}$ is the modulus with respect to $2\pi$.
\end{enumerate}
\end{theorem}

The ordering of $c_1,\ldots,c_n$ and $a_1,\ldots,a_n$ is not necessary and one can give permutation invariant expressions for these three densities as we do in~\eqref{1.1a},~\eqref{1.2a} and~\eqref{1.3a}, respectively. The restriction $\min_{j,k}|a_j-c_k|>0$ is a technical detail in the proof, but it is born out of the problematic limit $b\to0$ where the three expressions~\eqref{1.1},~\eqref{1.2} and~\eqref{1.3} seem to diverge. The interplay of the indicator function and the apparent singularity at $b=0$ is highly non-trivial and is reflected in a breakdown of our proof. Indeed, when $b\neq0$ or excluding a multiple of $2\pi$ for the third case then $a_j=c_k$ for some $j$ and $k$ describes only a set of measure zero which can be excluded without loss of generality.

The three cases of Theorem~\ref{prop} are proven in  sections~\ref{S2},~\ref{S3} and~\ref{S4}, respectively. Therein, we briefly recall  the spherical transform on ${\rm Herm}(n)$ and ${\rm Herm}_+(n)$ and introduce it for $\U(n)$. Specifically, in relation to the latter case, the spherical functions are the normalised irreducible characters of $\U(n)$ while $(\U(n)\times \U(n), \diag(\U(n)\times \U(n)))$ is a Gelfand pair. Yet, we are not aware that the explicit inversion formula of the spherical transform for this case has been reported in literature. The proof for the sum of two matrices is known~\cite{FG06,FZ19,Fa19} and shall serve as an illustration for the main steps of the proof for the other two cases. The proof of the known result~\eqref{1.3} is an alternative one given in~\cite{FZ19} where the Harish-Chandra--Itzykson--Zuber integral~\cite{HC57,IZ80} has been employed. In spite of this integral being related to harmonic analysis, it is not the same as that naturally encountered for products on ${\rm Herm}_+(n)$, which are the characters. The result~\eqref{1.2} is completely new.

Discretization of our approach to the eigenvalues of the sums and products of random matrix theory can be identified in the literature. Here we highlight the representation theoretical one pursued in~\cite{Vadim,Vadim2,Vadim3,Vadim4}, where various stochastic processes, like the GUE corner process, the Jacobi product process or the Schur process, have been studied, together with their continuum limits. Thus our generalizations of the Fourier transformation certainly encodes the decomposition of the distributions in irreducible representations of the corresponding group action; this becomes especially clear in the case of the multiplication on $\U(n)$. The approach in~\cite{Vadim,Vadim2,Vadim3,Vadim4} considers the problem of products and sums of matrices on the level of Schur polynomials which are the characters of the irreducible representations of $\U(n)$. Therefore, Theorem~\ref{prop} can be anticipated also from this point of view.

In brevity, we would like to mention our notation throughout the present work. The Lebesgue measure on a flat matrix space like the Hermitian matrices or diagonal matrices is denoted by $(\dv X)$ and is a product of all its real independent differentials, e.g. $(\dv X):=\prod_{1\le j\le k\le n}\dv x_{jk}^{(r)}\prod_{1\le j< k\le n}\dv x_{jk}^{(i)}$. In contrast, the normalized Haar measure such as of the unitary group $\U(n)$ or its cosets is denoted by $\mu(\dv U)$. Distributions of random variables are indicated by subscripts of the function such as $f_{X}$ for the random variable $X$. The space of absolutely integrable functions is denoted by $L^1$, while the space of sequences whose series is absolutely convergent is denoted by $l^1$.

\section{Additive Horn Problem on $\mathrm{Herm}(n)$}\label{S2}

The univariate Fourier transform for a real random variable $X$ with probability density $f_{X_1}(x)$ is the characteristic function \footnote{In probability theory, conventionally a characteristic function is defined as $\E_X[e^{itX}]$ instead.}
\begin{equation}
\E_{X_1}[e^{-itX_1}]=\int_\mathbb{R} f_{X_1}(x)e^{-itx_1}\dv x_1,
\end{equation}
where $\E[.]$ is the expectation value. When considering the sum of two independent random variables $X_1+X_2$, the characteristic function is the product of the two corresponding ones, i.e.,
\begin{equation}
\E_{X_1,X_2}[e^{-it(X_1+X_2)}]=\E_{X_1}[e^{-itX_1}]\E_{X_2}[e^{-itX_2}].
\end{equation}
Thence, the probability density of $X=X_1+X_2$ can be recovered by the inverse Fourier transform,
\begin{equation}
f_{X_1+X_2}(y)=\frac{1}{2\pi}\int_\mathbb{R}\dv t\, \E_{X_1}[e^{-itX_1}]\E_{X_2}[e^{-itX_2}]\,e^{ity}.
\end{equation}
These steps are at the heart of all three randomized Horn problems, where the Fourier transform is the spherical transform for the additive group of Hermitian matrices and the Fourier factor $e^{-itX}$ is known as the spherical function. In particular,  the additive Horn problem is simply replacing the real random variables by Hermitian random matrices. The only additional input that is not present in the univariate case is the reduction of the Fourier transform to the eigenvalues of the matrices. This will be outlined in subsection~\eqref{sec:Fourier} and applied to the Horn problem in subsection~\ref{sec:Horn-Herm}.

\subsection{Eigenvalue Fourier Transform}\label{sec:Fourier}

We consider a positive normalised $L^1$-function $f_X$ being the probability density of the random variable $X$. Its matrix Fourier transform is
\begin{equation}\label{2.1}
	\hat f_X(S):=\E_X\left[\exp(-i\tr XS)\right]=\int_{\mathrm{Herm}(n)} (\dv X)\, f_X(X)\exp(-i\tr XS),
\end{equation}
for all complex $n\times n$ matrix $S$ for which the integral exists. 
The normalisation is now reflected by $\hat f_X(0)=1$.
The inversion is particularly simple when its Fourier transform is an $L^1$-functions, too, then we can omit a regularization and it takes the form (e.g., see~\cite[Section 4.1 eq. (4.6)]{FZ19})
\begin{equation}\label{2.3}
	f_X(X)=\frac{1}{2^n\pi^{n^2}}\int_{\mathrm{Herm}(n)} (\dv S)\, \hat f_X(S)\exp(i\tr XS).
\end{equation}

The corresponding convolution theorem reads
\begin{equation}\label{2.2}
	\hat f_{X_1+X_2}(S)=\hat f_{X_1}(S)\cdot \hat f_{X_2}(S)
\end{equation}
for two independent random matrices $X_1,X_2\in{\rm Herm}(n)$. Thus, in combination with the inverse~\eqref{2.3} we have again
\begin{equation}\label{2.2.b}
	f_{X_1+X_2}(X)=\int_{{\rm Herm}(n)} \frac{(\dv S)}{2^n\pi^{n^2}}\,\hat f_X(S)\cdot \hat f_Y(S)\exp(i\tr XS).
\end{equation}

So far everything works analogously to the univariate case. Yet, when studying the eigenvalues we have to consider the behaviour of the Fourier transform under the adjoint action of $\U(n)$ on ${\rm Herm}(n)$.
A well-known, but crucial property of the Fourier transform is that invariance under unitary conjugation of the original distribution $f_X(X)=f_X(UXU^\dagger)$ carries over to $\hat f_X$. We will exploit this to diagonalise $X=U\mathbf{x}U^\dagger$ and $S=V\mathbf{s}V^\dagger$ with the eigenvalues $\mathbf{x}=\diag(x_1,\ldots,x_n)$ and $\mathbf{s}=\diag(s_1,\ldots,s_n)$.  To avoid over-counting, $U$ as well as $V$ need to be drawn from the quotient group $\U(n)/\U(1)^n$ and the eigenvalues $\mathbf{x}$ and $\mathbf{s}$ are ordered. However, we can relax both restrictions by properly normalising the integrals which we will do in the following. The measure $(\dv X)$ is, then, decomposed as (see e.g.~\cite[Eq. (1.27) with $\beta=2$]{forrester10})
\begin{equation}\label{2.4}
	(\dv X)=\frac{\pi^{n(n-1)/2}}{\prod_{j=0}^{n}j!}\Delta(\mathbf x)^2(\dv \mathbf x)\mu(\dv U).
\end{equation}
Plugging~\eqref{2.4} into~\eqref{2.1} yields
\begin{equation}\label{2.5}
	\hat f_X(S)=\frac{\pi^{n(n-1)/2}}{\prod_{j=0}^{n}j!}\int_{\mathbb{R}^n}(\dv\mathbf x)\,f_X(\mathbf x)\Delta(\mathbf x)^2\left(\int_{\U(n)} \mu(\dv U)\, \exp(-i\tr U\mathbf xU^\dagger \mathbf s)\right).
\end{equation}
We encounter here the Harish-Chandra--Itzykson--Zuber (HCIZ) integral~\cite{HC57,IZ80},
\begin{equation}\label{2.6}
	\int_{\U(n)}\mu(\dv U)\, \exp\left(\tr U\mathbf xU^\dagger \mathbf s\right)=\prod_{j=0}^{n-1}j!\,\frac{\det\left[e^{x_js_k}\right]_{j,k=1}^n}{\Delta(\mathbf x)\Delta(\mathbf s)}:=\phi(\mathbf x,\mathbf s),
\end{equation}
which is the spherical function for the Fourier transform on the eigenvalues. It reflects the non-trivial metric of the induced space that originates from the eigenvalue decomposition. The corresponding spherical transform is obtained by substituting~\eqref{2.6} into~\eqref{2.5},
\begin{equation}\label{2.7}
	\hat f_X(\mathbf s)=\frac{\pi^{n(n-1)/2}}{\prod_{j=0}^{n}j!}\int_{\mathbb{R}^n}(\dv\mathbf x)\,\Delta(\mathbf x)^2f_X(\mathbf x)\phi(-i\mathbf x,\mathbf s).
\end{equation}
We denote $\s f_X:=\hat f_X(\mathbf s)$ as the spherical transform of $f_X$. Since also the HCIZ-integral $\phi$ is invariant under $\U(n)$ we can also understand it as a function of the full random matrix $X$, in particular $\phi(\mathbf x,\mathbf s)=\phi(X,\mathbf s)$. This implies that the Fourier (spherical) transform has also the form $\hat f_X(\mathbf s)=\s f_X(s)=\E[\phi(-iX,\mathbf s)]$. Defining the transform in this way allows us to relax the unitarily invariance of $f_X$ and can even compute the spherical transform of fixed matrices which is simply the spherical function, meaning the HCIZ-integral in the current case.

Due to $\hat f_X(\mathbf s)=\hat f_X(S)$, the convolution theorem~\eqref{2.2} still holds
\begin{equation}\label{2.8}
	\s f_{X+Y}=\s f_{X}\cdot \s f_{Y}
\end{equation}
and, also, the inverse carries over and can be readily written as 
\begin{equation}\label{2.9}
	\s^{-1}(\s f_X)(\mathbf x):=\frac{1}{\pi^{n(n-1)/2}\prod_{j=0}^{n}j!}\int_{\mathbb{R}^n}\frac{(\dv\mathbf s)}{(2\pi)^n}\,\Delta(\mathbf s)^2\s f_X(\mathbf s)\phi(i\mathbf x,\mathbf s),
\end{equation}
with $\s f_X\Delta(\mathbf s)^2$ being an $L^1$-function on $\mathbb{R}^n$, otherwise we need a regularization. Note that this inverse maps the spherical transform $\s f_X$ back to the matrix density $f_X$. To obtain the joint eigenvalue density, we need to multiply Eq.~\eqref{2.9} with $\Delta(\mathbf c)^2$ and the constant in~\eqref{2.4}.

\subsection{Randomised Additive Horn Problem and Proof of~\eqref{1.1}}\label{sec:Horn-Herm}

Before coming to the setting of 1.~in Theorem~\ref{prop}, we would like to study the general case of two fixed diagonal real matrices $\mathbf a=\diag( a_1,\ldots, a_n)$ 
and $\mathbf b=\diag( b_1,\ldots, b_n)$, each conjugated by an independent Haar distributed unitary matrix.
The result of the corresponding Horn problem in terms of an $n$-fold integral is given in various works~\cite{Fa15,Zu18,FZ19}. Nevertheless, we would like to outline the structure of the derivation since it is important for our subsequent development.

As a first ingredient, we need a convolution theorem involving fixed matrices. This is given in the following lemma.

\begin{lemma}[Convolution Theorem with Fixed Matrices]\label{l1}
	Let $X$ be a random matrix in $\mathrm{Herm}(n)$ distributed as $f_X$, let $\mathbf x_0$ be a fixed real diagonal matrix, and $U$ be a Haar distributed $\U(n)$ matrix. The spherical transform of the sum $X+U\mathbf x_0U^{\dagger}$ is
	\begin{equation}
	\s f_{X+U\mathbf x_0U^{\dagger}}(\mathbf s)=\s f_X(\mathbf s)\cdot\phi(-i\mathbf x_0,\mathbf s).
	\end{equation}
\end{lemma}

\begin{proof}
	This is a direct consequence of the Fourier transform~\eqref{2.1} and the independence of $X$ and $U$, i.e.
	\begin{align}
	\E_{X,U}[\exp(-i\tr (X+U\mathbf x_0U^\dagger)S)]&=\E_{X}[\exp(-i\tr XS)]\E_{U}[\exp(-i\tr U\mathbf x_0U^\dagger S)].
	\end{align}
	The second term is the HCIZ integral~\eqref{2.6}.
\end{proof}

For a valid inversion formula $\s^{-1}\s f=f$, we will encounter a point where we need to interchange two integrals, which requires a Fubini's theorem. And to satisfy the condition of Fubini's theorem, it is advisable to introduce an auxiliary random matrix $H_\varepsilon$ ($\varepsilon>0$) as a regularisation with the density and spherical transform (see~\cite[\S 3.1]{KR17})
\begin{equation}\label{2.10}
f_{H_\varepsilon}(x_1,\cdots, x_n)=\prod_{j=1}^n\frac{e^{-x_j^2/4\varepsilon^2}}{2\sqrt{\pi}\varepsilon},\qquad\s f_{H_\varepsilon}(s_1,\cdots, s_n)=\prod_{j=1}^ne^{-\varepsilon^2s_j^2},
\end{equation}
which is the Gaussian unitary ensemble (GUE).
In this way, the density as well as the spherical transform of the sum $C_\varepsilon=U\mathbf aU^\dagger +V\mathbf bV^\dagger+H_\varepsilon$ are always guaranteed to be $L^1$-functions.

When applying Lemma~\ref{l1} twice, the spherical transform of $C_\epsilon$ is
\begin{equation}\label{2.10.b}
\s f_{C_\varepsilon}(s_1,\cdots, s_n)=\left(\prod_{j=1}^ne^{-\varepsilon^2s_j^2}\right)\,\phi(-i\mathbf a,\mathbf s)\phi(-i\mathbf b,\mathbf s).
\end{equation}
Next, we exploit~\eqref{2.9} and then take the limit $\varepsilon\rightarrow0$ to obtain the joint eigenvalue density of $C$,
\begin{equation}\label{2.11}
	p(\mathbf c)=\frac{\Delta(\mathbf c)^2}{\prod_{j=0}^{n}(j!)^2}\lim_{\varepsilon\rightarrow 0}\int_{ \mathbb{R}^n} \frac{(\dv \mathbf s)}{(2\pi)^n}\,\Delta(\mathbf s)^2\left(\prod_{j=1}^ne^{-\varepsilon^2s_j^2}\right)\phi(-i\mathbf a,\mathbf s)\phi(-i\mathbf b,\mathbf s)\phi(i\mathbf c,\mathbf s);
\end{equation}
Here $\prod_{j=1}^ne^{-\varepsilon^2s_j^2}$ is the additional regularisation. See~\cite{Zu18},~\cite[eq. (5)]{Coquereaux2017} and its generalisation in~\cite{CRZ20}.

The particular case where  we have $\mathbf b=\diag(b,0,\cdots,0)$ has rank one is given in~\cite[\S 4.1 in particular Eqn. (4.10)]{FZ19} and \cite[Thm. 4.1]{Fa15}. The HCIZ-integral~\eqref{2.6} simplifies in this case to
\begin{equation}\label{2.12}
\phi(-i\mathbf b,\mathbf s)=\frac{(-i)^{n-1}(n-1)!}{b^{n-1}}\sum_{p=1}^n\frac{e^{-ibs_p}}{\prod_{l\ne p}(s_l-s_p)},
\end{equation}
allowing (\ref{2.11}) to be written 
\begin{equation}\label{2.13}
\begin{split}
p(\mathbf c)=&\frac{(-i)^{n-1}(n-1)!}{(n!)^2b^{n-1}}\frac{\Delta(\mathbf c)}{\Delta(\mathbf a)}\lim_{\varepsilon\rightarrow 0}\int_{\mathbb{R}^n} \frac{(\dv \mathbf s)}{(2\pi)^n}\left(\prod_{j=1}^ne^{-\varepsilon^2s_j^2}\right)\\
&\times\det[e^{ic_js_k}]_{j,k=1}^n\det[e^{-ia_js_k}]_{j,k=1}^n\sum_{p=1}^n\frac{e^{-ibs_p}}{\prod_{l=1,l\ne p}^n(s_l-s_p)}.
\end{split}
\end{equation}
This integral serves as the starting point of the proof of Eq.~\eqref{1.1}. It will be carried out in two steps where we first massage the integrals and finally perform the limit $\varepsilon\to0$.

\subsubsection*{Step 1: Application of Andr\'eief's Identity}

Due to the sum in~\eqref{2.13}, we cannot apply Andr\'eief's identity~\cite{An86} directly. We can try instead the generalized version of it derived in~\cite[Appendix C]{KG10}. For this aim, it is helpful to notice that the poles at $s_l=s_p$ in the sum are all cancelled by the zeros of the two determinants. Furthermore, each summand yields the same contribution due to permutation symmetry of the integrand so that we can select one, say $p=1$. When integrating over $s_2,\ldots,s_n$ first and shifting the contour of $s_1$ by an imaginary increment $+iy$ with $y>0$, we can apply the generalized Andr\'eief identity~\cite[Eqn.~(C.4)]{KG10}  and arrive at
\begin{equation}\label{2.14}
\begin{split}
p(\mathbf c)=&-\frac{1}{n\,b^{n-1}}\frac{\Delta(\mathbf c)}{\Delta(\mathbf a)}\lim_{\varepsilon\rightarrow 0}\int_\mathbb{R}\dv s_1\, e^{-\varepsilon^2(s_1+iy)^2-ib(s_1+iy)}\\
&\times\det\begin{bmatrix}
0&\displaystyle\left[e^{ic_k(s_1+iy)}\right]_{k=1,\cdots,n}
\\\\
\displaystyle\left[e^{-ia_j(s_1+iy)}\right]_{j=1,\cdots,n}&\displaystyle\left[\int_{\mathbb{R}}\frac{e^{-\varepsilon^2s^2+i(c_k-a_j)s}}{s-(s_1+iy)}\frac{\dv s}{2\pi i}\right]_{j,k=1}^n
\end{bmatrix}.
\end{split}
\end{equation}
The shift in $+iy$ is controlled by the Gaussian regularization and guarantees that the integral in the determinant is still absolutely integrable. 
This integral is a product of two $L^2$-functions, so that we can make use of Plancherel's theorem,
\begin{equation}\label{2.23}
\int_{\mathbb{R}}\frac{e^{-\varepsilon^2s^2+i(c_k-a_j)s}}{s-(s_1+iy)}\frac{\dv s}{2\pi i}
=e^{i(c_k-a_j)(s_1+iy)}\int_{(a_j-c_k)/(2\varepsilon)}^\infty \frac{e^{-u^2}}{\sqrt{\pi}}e^{2i\varepsilon u(s_1+iy)}\dv u.
\end{equation}
The integral is essentially the complimentary error function ${\rm erfc}$ which is analytic allowing us to remove the shift $iy\to0$. Pulling out common terms, the density becomes
\begin{equation}\label{2.27}
\begin{split}
p(\mathbf c)=&-\frac{1}{n\,(2b)^{n-1}}\frac{\Delta(\mathbf c)}{\Delta(\mathbf a)}\lim_{\varepsilon\rightarrow 0}\int_\mathbb{R}\dv s_1\, \exp\left[-n\varepsilon^2s_1^2-i\left(b+\sum_{j=1}^n(c_j-a_j)\right)s_1\right]\\
&\times\det\begin{bmatrix}
0&\displaystyle\left[1\right]_{1\times n}
\\\\
\left[1\right]_{n\times 1}&
\displaystyle\left[{\rm erfc}\left[\frac{a_j-c_k}{2\varepsilon}-i\varepsilon s_1\right]\right]_{j,k=1}^n
\end{bmatrix}.
\end{split}
\end{equation}
Here, $[1]_{p\times q}$ denotes a $p\times q$ matrix with all its entries being $1$.

\subsubsection*{Step 2: Limit $\varepsilon\to0$}

Our goal is to replace the above error function by the Heaviside step function in the $\varepsilon\rightarrow 0$ limit, where we need the condition $\min_{j,k}|c_k-a_j|>0$. We can firstly make an estimate
\begin{equation}\label{30}
\begin{split}
\left|{\rm erfc}\left[\frac{a_j-c_k}{2\varepsilon}-i\varepsilon s_1\right]-2\Theta(c_k-a_j)\right|
=&\left|{\rm erfc}\left[\frac{|a_j-c_k|}{2\varepsilon}-i\varepsilon\,{\rm sign}(a_j-c_k) s_1\right]\right|\\
<&\frac{\exp[-(a_j-c_k)^2/4\varepsilon^2+\varepsilon^2 s_1^2]}{|(a_j-c_k)/\varepsilon-i\varepsilon s_1|}
\end{split}
\end{equation}
with $\Theta$ the Heaviside step function. We furthermore denote the difference between the error function ${\rm erfc}\left[\frac{a_j-c_k}{2\varepsilon}-i\varepsilon s_1\right]$ and $2\Theta(c_k-a_j)$ by $e_{jk}$. Then, $e_{jk}$ is of the order or smaller than $\mathcal{O}(\varepsilon e^{-(a_j-c_k)^2/4\varepsilon^2+\varepsilon^2 s_1^2})$. 

Now we replace the error function in ~\eqref{2.27} by
\begin{equation}
	e_{jk}+2\Theta(c_k-a_j)
\end{equation}
and expand the determinant. This will give us terms of products involving at least one $e_{jk}$. Then with a rescaling $s_1=y/\varepsilon$, we look at the integrals that involve the form
\begin{equation}\label{2.28}
\int_\mathbb{R} e^{-(n-l)y^2}\exp\left(\frac{iy}{\varepsilon}(-b+\sum_{j=1}^n(c_j-a_j)) \right)\, \mathcal O(\varepsilon^{l-1}e^{-l\min_{j,k}|a_j-c_k|^2/4\varepsilon^2})\dv y
\end{equation}
for $l=1,\ldots,n-1$, which all vanish  exponentially at least like $e^{-\min_{j,k}|a_j-c_k|^2/4\varepsilon^2}$. The order shows the leading order of these functions in limit. This tells us that all terms involving at least one $e_{jk}$ vanish in the limit and therefore, only the contribution $2\Theta(c_k-a_j)$ in the expansion survives, which leads to the desired determinant
\begin{equation}
	\det\begin{bmatrix}
	0&\displaystyle[1]_{1\times n}
	\\
	\displaystyle[1]_{n\times 1}&\displaystyle\left[\Theta(c_j-a_k)\right]_{j,k=1}^n
	\end{bmatrix}.
\end{equation}

This procedure decouples the corresponding integral in $s_1$ which is a Gaussian. The remaining limit is only given in a weak topology creating the Dirac delta function that reflects the trace condition~\eqref{1.01}, so that we eventually arrive at
\begin{equation}\label{1.1a}
p(\mathbf c)=-\frac{1}{n\,b^{n-1}}\frac{\Delta(\mathbf c)}{\Delta(\mathbf a)}\,\delta\left(b+\sum_{j=1}^{n}(a_j-c_j)\right)\det\begin{bmatrix}
0&\displaystyle[1]_{1\times n}
\\
\displaystyle[1]_{n\times 1}&\displaystyle\left[\Theta(c_j-a_k)\right]_{j,k=1}^n
\end{bmatrix}.
\end{equation}
The determinant with the minus sign is equivalent with the ordering in~\eqref{1.1} which appears $n!$ times in the above permutation invariant version. This concludes the proof.

\section{Multiplicative Horn Problem on $\mathrm{Herm}_+(n)$}\label{S3}

In analogy to the sum of matrices, we can ask what is the corresponding univariate case for products of positive definite Hermitian matrices, in particular what is the corresponding transform that factorises the problem. For products of random variables on the positive real line this is the Mellin transform
\begin{equation}\label{3.01}
\mathbb E_{X_1}\left[X_1^{s-1}\right]:=\int_{\mathbb{R}_+}f_{X_1}(x)\,x^{s}\,\frac{\dv x}{x}, 
\end{equation}
for all $s\in\mathbb{C}$ where the integral exists for the probability density $f_{X_1}$.
Here, $\dv x/x$ is recognized as the Haar measure on the multiplicative group $\mathbb{R}_+$. The corresponding convolution theorem reads
\begin{equation}\label{3.02}
\mathbb E_{X_1,X_2}\left[(X_1X_2)^{s-1}\right]=\mathbb E_{X_1}\left[X_1^{s-1}\right]\cdot \mathbb E_{Y_1}\left[Y_1^{s-1}\right].
\end{equation}
The matrix analogue on ${\rm Herm}_+(n)$ is the spherical transform~\cite{He84,KK16} which we will briefly recall in subsection~\ref{sec:spher} and will be applied to the multiplicative Horn problem in subsection~\ref{sec:Horn-Pos}.

\subsection{Spherical Transform on $\mathrm{Herm}_+(n)$}\label{sec:spher}

To mimic the Mellin transform~\eqref{3.01} on the matrix level we need to say first and foremost what is the generalisation of $x^s$. For an $X\in\mathrm{Herm}_+(n)$ and $\mathbf s=(s_1,\cdots,s_n)\in\C^n$, this function is Selberg's generalised power function~\cite{Se56}
\begin{equation}\label{Selberg}
|X|^{\mathbf s}:=\prod_{j=1}^{n-1}\det(X_{j\times j})^{s_j-s_{j+1}-1}\cdot\det(X)^{s_n}
\end{equation}
where $X_{j\times j}$ is the $j\times j$ upper left block of the matrix $X$. This suggests choosing
\begin{equation}\label{3.1}
\hat f_X(\mathbf s):=\int_{\mathrm{Herm}_+(n)}f_X(X)\,|X|^{\mathbf s}\,\frac{(\dv X)}{(\det X)^n}
\end{equation}
as the multivariate generalisation of~\eqref{3.01} with $(\dv X)/(\det X)^n$ as the Haar measure with respect to the closed multiplication $(X_1,X_2)\mapsto X_1^{1/2}X_2X_1^{1/2}$ on ${\rm Herm}_+(n)$ (observe that $X_1^{1/2}X_2X_1^{1/2}$ and $X_1X_2$ have the same eigenvalues, but the latter is no longer in $Herm_+(n)$). Indeed for $n=1$, \eqref{3.1} reduces to~\eqref{3.01}.

There is, however, a critical problem which forbids an inverse of the transform, namely we mapped a function of $n^2$ variables to one on $n$. Thus we need to restrict its definition to the eigenvalue statistics. When assuming that $f_{X}$ is unitarily invariant again, we are allowed to diagonalise $X=Ue^{\mathbf{x} }U^\dagger$ and integrating over $U$ alone, which leads to the Gelfand-Na\u\i mark integral~\cite{GN57}
\begin{equation}\label{3.3}
	\int_{\U(n)}\mu(\dv U)\,|Ue^{\mathbf{x} }U^\dagger|^{\mathbf s}=\prod_{j=0}^{n-1}j!\,\frac{\det[e^{x_js_k}]_{j,k=1}^n}{\Delta(e^\mathbf x)\Delta(\mathbf s)}:=\phi(e^\mathbf x,\mathbf s).
\end{equation}
This integral is the counterpart of the HCIZ-integral~\eqref{2.6} and they are exactly the spherical functions. The corresponding spherical transform is (see e.g. \cite[\S 2.4]{KK16}, \cite{KK19,KFI17,Te16})
\begin{equation}\label{3.4}
\hat f_{X}(\mathbf s)=\frac{\pi^{n(n-1)/2}}{\prod_{j=0}^{n}j!}\int_{\mathbb{R}^n}\frac{(\dv \mathbf x)}{\exp({(n-1)\sum_{j=1}^n x_j})}\Delta(e^\mathbf x)^2 f_{X}(e^\mathbf x) \phi(e^\mathbf x,\mathbf s):=\s f_{X}(\mathbf s).
\end{equation}
Here the factor $\exp({-(n-1)\sum_{j=1}^n x_j})$ is the combination of the Haar measure 
\linebreak $\det X^{-n}\dv X$ and the Jacobian arising from the change of variables $x_j\mapsto e^{x_j}$. The normalization reads $\hat f_X(\mathbf{s}_0)=1$ with $\mathbf{s}_0=\diag(0,\ldots,n-1)+n\eins_n$, where the the term $n \mathbf 1_n$ yields a factor $\det X^n$, which effectively
cancels with the determinant factor in the Haar measure.

In general the spherical transform is defined by $\hat f_X(\mathbf s)=\E_{X}[\phi(X,\mathbf s)]$ which also holds for functions that are not unitarily invariant, where $\phi$ is then given by the left hand side of~\eqref{3.3} with $e^{\mathbf{x}}$ replaced by $X$. Indeed for random matrices invariant under the adjoint action of the unitary group it agrees with the above definition. In this way, the spherical transform of a fixed matrix $X$ as well as of the product $UXU^\dagger$ with $U$ being a Haar distributed unitary matrix is only the spherical function divide by $(\det X)^n$. 

The multiplicative counterpart of the convolution theorem~\cite{He84} is
\begin{equation}\label{3.5}
\s f_{X_1^{1/2}X_2X_1^{1/2}}=\s f_{X_1}\cdot\s f_{X_2}.
\end{equation}
This, however, only holds true when one of the random matrices $X_1$ or $X_2$ is unitarily invariant. This formula together with the inversion~\cite{He84,KK16}
\begin{equation}\label{3.6}
	\s^{-1}(\s f_X)(e^\mathbf x)=\frac{(-1)^{n(n-1)/2}}{\pi^{n(n-1)/2}\prod_{j=0}^{n}j!}\int_{\mathbb{R}^n} \frac{(\dv \mathbf s)}{(2\pi )^n}\Delta(\mathbf{s}_0+i\mathbf s)^2 \phi(e^{-\mathbf x},\mathbf{s}_0+i\mathbf s)\s f_X(\mathbf{s}_0+i\mathbf s)
\end{equation}
 highlights that the eigenvalue statistics of  $X_1^{1/2}X_2X_1^{1/2}$ and $X_2^{1/2}X_1X_2^{1/2}$ are the same. For the inverse, we  have again suppressed the regularization in~\eqref{3.6} in contrast to~\cite[(2.40)]{KK16}. Thus, we assume that $\Delta(\mathbf{s}_0+i\mathbf s)^2 \s f_X(\mathbf{s}_0+i\mathbf s)$ is an $L^1$-function on $\mathbb{R}^n$.

The definitions and properties above can be also established in the group and representational theoretical language in~\cite[Chapter IV \S 2.1]{He84}. From that perspective it does not come as a surprise that the spherical functions resemble the Schur polynomials, meaning the characters of the irreducible representations of $\U(n)$. Yet, there is a subtle difference. While for the characters $s$ has to be an array of integers, here it is a complex vector, cf. Sec.~\ref{S4}.

\subsection{Randomised Multiplicative Horn Problem on $\mathrm{Herm}_+(n)$ and Proof of~\eqref{1.2}}\label{sec:Horn-Pos}

Analogous to the sum, we need to establish the convolution that involve fixed matrices.

\begin{lemma}[Multiplicative Convolution with Fixed Matrices on ${\rm Herm}_+(n)$]\label{l4}
	Let $X$ be a random matrix in $\mathrm{Herm}_+(n)$, $\mathbf x_0$ be a real diagonal matrix, and $U$ be a Haar distributed unitary matrix. Then, the spherical transform of the random matrix $X^{1/2}Ue^{\mathbf x_0}U^\dagger X^{1/2}$ is
	\begin{equation}
	\s f_{X^{1/2}Ue^{\mathbf x_0}U^\dagger X^{1/2}}(\mathbf s)=\s f_X(\mathbf s)\cdot e^{-n\tr\, \mathbf x_0}\phi(e^{\mathbf x_0},\mathbf s).
	\end{equation}
\end{lemma}

\begin{proof}
	The result is a direct consequence of the integral~\cite{KK16,He84}
	\begin{equation}
	\int_{\U(n)}\mu(\dv U)\phi(X_1^{1/2}UX_2U^\dagger X_1^{1/2},\mathbf s)=\phi(X_1,\mathbf s)\phi(X_2,\mathbf s)
	\end{equation}
	which holds for any $X_1,X_2\in{\rm Herm}_+(n)$ and $\mathbf{s}\in\mathbb{C}$.
\end{proof}

As before, we need to introduce an auxiliary unitarily-invariant random matrix $H_\varepsilon$ with eigenvalues $e^\mathbf x$, whose joint eigenvalue density is
\begin{align}\label{3.10}
f_{H\varepsilon}(e^\mathbf x)=\frac{1}{\mathcal N_\varepsilon}\Delta(e^\mathbf x)\det\left[\left(-\frac{\partial}{\partial x_k}\right)^{j-1}\frac{e^{-x_k^2/4\varepsilon^2}}{2\sqrt{\pi}\varepsilon}\right]_{j,k=1}^n\propto\Delta(\mathbf x)\Delta(e^{\mathbf x})e^{-\tr\,\mathbf{x}^2/4\varepsilon^2}.
\end{align}
This relation (in particular, the Vandermonde determinant $\Delta(\mathbf x)$) follows by use of the Rodrigues formula for the Hermite polynomials
to compute the derivatives, removal of the remaining Gaussian from each column, then noting the remaining determinant of Hermite polynomials is proportional to $\Delta(\mathbf x)$ as is standard.

This ensemble is one of the Muttalib-Borodin ensembles~\cite{Muttalib1995,Borodin1998,FW15,KK19}, which is a special case of a P\'olya distribution with a log-normal weight~\cite[Example 2.4 (a)]{FKK}. It is the solution of the DMPK equation~\cite{Dorokhov1982,Mello1988,Frahm1995,Beenakker1997} and represents the heat kernel of the multiplicative Dyson-Brownian motion on the positive Hermitian matrices 
$\mathrm{Herm}_+(n)=\mathrm{GL}(n,\C)/\U(n)$ ~\cite{Ipsen2016}. Hence, it is well-known that it describes a probability density with the normalisation constant
		\begin{equation}
		\mathcal N_\varepsilon=n!\prod_{j=0}^{n-1}j!e^{(j+1)^2\varepsilon^2},
		\end{equation}
		and that its spherical transform reads~\cite[Example 3.6 (e)]{KK19}
		\begin{equation}\label{3.15x}
		\s f_{H\varepsilon} (\mathbf s)=\prod_{j=1}^ne^{\varepsilon^2(s_j-n+1)^2-\varepsilon^2j^2}.
		\end{equation}

We will combine the regularization $H_\varepsilon$ with the fixed diagonal matrices $e^\mathbf a=(e^{a_1},\cdots,e^{a_n})$ and $e^\mathbf b=(e^{b_1},\cdots,e^{b_n})$ via the Haar distributed unitary matrices $U,V \in \U(n)$ to
form the matrix
 $$C_\varepsilon=H_\varepsilon^{1/2}(UAU^\dagger)^{1/2}VBV^\dagger(UAU^\dagger)^{1/2}H_\varepsilon^{1/2}.$$
The spherical transform of $C_\varepsilon$ is then
\begin{equation}\label{3.14}
	\s f_{C_\varepsilon}(\mathbf s)=\s f_{H_\varepsilon}(\mathbf s)\cdot \phi(e^\mathbf a,\mathbf s)\cdot \phi(e^\mathbf b,\mathbf s),
\end{equation}
and,  with the help of~\eqref{3.6}, its joint eigenvalue density becomes
\begin{equation}\label{3.112}
\begin{split}
	p(e^\mathbf c)=&\frac{1}{\prod_{j=0}^{n}(j!)^2}\,\frac{\Delta( e^\mathbf c)^2}{e^{n\sum_{j=1}^n(a_j+b_j)}}\lim_{\varepsilon\to 0}\int_{\mathbb{R}^n} \frac{(\dv \mathbf s)}{(2\pi )^n}\Delta(\mathbf{s}_0+i\mathbf s)^2\\
	&\times \left(\prod_{j=1}^ne^{\varepsilon^2(j+i s_j)^2-\varepsilon^2j^2}\right)\,\phi(e^\mathbf a,\mathbf{s}_0+i\mathbf s)\,\phi(e^\mathbf b,\mathbf{s}_0+i\mathbf s)\,\phi(e^{-\mathbf c},\mathbf{s}_0+i\mathbf s).
	\end{split}
\end{equation}
This is the analogue of~\eqref{2.11} and shows the main problem to overcome. In the general setting, we have too many Vandermonde determinants of $\mathbf{s}$ in the denominator to deal with them analytically, as in the case for the additive Horn problem. Therefore, we will specialize to the case that $e^\mathbf b$ is rank-1.

Thus choosing $\mathbf b=(b,0,\cdots,0)$, we need to know the corresponding spherical function $\phi(e^\mathbf b,\mathbf{s})$. The computation works analogous to the additive case by repetitively applying l'H\^opital's rule to~\eqref{3.3}, which yields
	\begin{equation}\label{3.23}
	\begin{split}
	\phi(e^\mathbf b,\mathbf s)=&\frac{(n-1)!}{(1-e^b)^{n-1}}\det
	\begin{bmatrix}
	\left[e^{bs_k}\right]_{k=1,\cdots,n}\\
	\left[s_{k}^{j-1}\right]_{\substack{j=1,\cdots, n-1\\k=1,\cdots,n}}
	\end{bmatrix}=\frac{(n-1)!}{(1-e^b)^{n-1}}\sum_{p=1}^n\frac{e^{bs_p}}{\prod_{\substack{l=1\\l\ne p}}^n(s_l-s_p)}.
	\end{split}
	\end{equation}
This is essentially the same computation as for~\eqref{2.12}.

Therefore, the density~\eqref{3.112} explicitly becomes
\begin{equation}\label{3.18}
\begin{split}
	p(e^{\mathbf c})=&\frac{(n-1)!}{(n!)^2(1-e^{b})^{n-1}}\frac{\Delta({e^\mathbf c})}{\Delta({e^\mathbf a})}\lim_{\varepsilon\rightarrow 0}\int_{\mathbb{R}^n}\frac{(\dv \mathbf s)}{(2\pi)^n}\sum_{p=1}^n\frac{e^{b(p-1+is_p)}}{\prod_{\substack{l=1\\l\ne p}}^n[l-p+i(s_l-s_p)]}\\
	&\times \left(\prod_{j=1}^ne^{\varepsilon^2(j+i s_j)^2-\varepsilon^2j^2}\right)\det[e^{-c_j(k-1+is_k)}]_{j,k=1}^n\det[e^{a_j(k-1+is_k)}]_{j,k=1}^n.
\end{split}
\end{equation}
Due to the Gaussian regularization we do not need the shift by $\mathbf{s}_0$ originally given in the inversion formula~\eqref{3.6}. Changing variables $s_j=-y_j+i j$ ($j=1,\cdots,n$) shows
\begin{equation}\label{3.18.b}
\begin{split}
	p(e^{\mathbf c})=&\frac{(-i)^{n-1}(n-1)!e^{\sum_{j=1}^n(a_j-c_j)-b}}{(n!)^2(e^{b}-1)^{n-1}}\frac{\Delta({e^\mathbf c})}{\Delta({e^\mathbf a})}\lim_{\varepsilon\rightarrow 0}\int_{\mathbb{R}^n}\frac{(\dv \mathbf y)}{(2\pi)^n}\sum_{p=1}^n\frac{e^{-iby_p}}{\prod_{\substack{l=1\\l\ne p}}^n(y_l-y_p)}\\
	&\times \left(\prod_{j=1}^ne^{-\varepsilon^2y_j^2}\right)\det[e^{ic_jy_k}]_{j,k=1}^n\det[e^{-ia_jy_k}]_{j,k=1}^n.
\end{split}
\end{equation}
This $\mathbf y$-integral is exactly the same as in~\eqref{2.13}, telling us that
\begin{equation}\label{1.2a}
p(e^{\mathbf c})=-\frac{1}{n(e^b-1)^{n-1}}\frac{\Delta(e^\mathbf c)}{\Delta(e^\mathbf a)}\, \delta\left(b+\sum_{j=1}^n(a_j-c_j)\right)\det\begin{bmatrix}
0&\displaystyle[1]_{1\times n}
\\
\displaystyle[1]_{n\times 1}&\displaystyle\left[\Theta(c_j-a_k)\right]_{j,k=1}^n\end{bmatrix},
\end{equation}
or its equivalent form~\eqref{1.2}, which concludes the proof.

\begin{remark}
In~\cite{FZ19} an algebraic method has been used to derive~\eqref{1.1} and~\eqref{1.3}, which can be carried over to this case. The matrix $C$ has the same eigenvalue as $A^{1/2}WBW^\dagger A^{1/2}$ where $W=U^\dagger V\in\U(n)$ drawn from the Haar measure. In the rank-$1$ case, the characteristic polynomial of $C$ is
\begin{align}
	\det(\lambda I-A^{1/2}WBW^\dagger A^{1/2})&=\det(\lambda I-AWBW^\dagger)\nonumber\\
	&=\det(\lambda I-A)(b-(b-1) \lambda\mathbf w^\dagger(\lambda-A)^{-1}\mathbf w\mathbf ),
\end{align}
where $\mathbf w$ denotes the first column of $W$. Then the joint eigenvalue density~\eqref{1.2} can be deduced by the same discussion as in~\cite[Prop. 2]{FZ19}, except that $C$ now has eigenvalues parametrised by $e^\mathbf c$.
\end{remark}

\begin{remark}\label{rmk11}
	In other works (e.g.~\cite{BCDL,Kl00}) a different version of multiplicative Horn problem relating to the singular values of products of matrices is considered. The corresponding randomised version is to find the singular value probability density of the matrix $UXU^\dagger\cdot VYV^\dagger$ while $X, Y$ are two $n\times n$ fixed real diagonal matrices. This can be formulated into our setting, as the square of the desired singular values are the eigenvalues of the matrix $UXU^\dagger\cdot VY^2V^\dagger\cdot UXU^\dagger$. It is then clear that by letting $A=X^2$ and $B=Y^2$, Eq.~\eqref{1.2} is the joint probability density of the squared singular values in the case where $Y$ has rank one.
\end{remark}

\section{Multiplicative Horn problem on $\mathrm{U}(n)$}\label{S4}

When it comes to harmonic analysis on $\U(n)$, the first thing to do is to seek a univariate analogue which is the multiplication of complex, unimodular phases, i.e., the group $\U(1)$. For a random variable $X$ in $\U(1)$ with probability density $f_X(z)\in L^1(S^1)$ the corresponding Mellin transform is equivalent to the Fourier transform on the interval $(-\pi,\pi]$,
\begin{equation}\label{4.0}
\E_X[X^{s-1}]=\int_{\U(1)}f_X(z)z^{s}\mu(\dv z)=\int_{(-\pi,\pi]}f_X(e^{i\theta})e^{is\theta }\frac{\dv \theta}{2\pi}
\end{equation}
with $s\in\Z$ and $\mu(\dv z)$ the normalised Haar measure on $\U(1)$. The subtlety of choosing integers for $s$, instead of the complex numbers as used for the other two cases, is crucial. First and foremost, taking a root of a complex number is not unique and one has to select a cut in the complex plane, destroying its continuity as a function. Secondly, the expectation value $\E_X[X^{s-1}]$ is up to a constant equal to the $s$--th Fourier component of $f_X$,  and so knowledge of $\{\E_X[X^{s-1}] \}_{s\in\Z}$ suffices to reconstruct $f_X$ via a Fourier series. Thus, the spherical function in this case is the complex phase $z=e^{i\theta}$.

Since the Mellin transform on $\U(1)$ has the same form~\eqref{3.01} as defined on $\mathbb{R}_+$, it permits the same multiplicative convolution theorem~\eqref{3.02}, or equivalently an additive convolution theorem in terms of the variable $\theta$. We will extend this idea to the general case $\U(n)$ in subsection~\ref{sec:unitary} and apply it to the Horn problem in subsection~\ref{sec:Horn.unitary}.

\subsection{Spherical transform on $\mathrm{U}(n)$}\label{sec:unitary}

The spherical functions $z^s$ with $s\in\Z$ are also known as the characters of the irreducible representations of $\U(1)$. The most natural generalisation to $\U(n)$ is therefore the corresponding normalized characters
\begin{equation}\label{4.4.1}
\phi(X,\mathbf{s})=\frac{{\rm ch}_{\mathbf{s}}(X)}{{\rm ch}_{\mathbf{s}}(\eins_n)},
\end{equation}
where  $X\in\U(n)$ and $\mathbf{s}\in\Z^n$ is a partition that corresponds to an irreducible representation  of $\U(n)$. Indeed these normalised characters satisfy all the necessary conditions needed for 
an harmonic analysis on $\mathrm{U}(n)$.

In particular, the choice (\ref{4.4.1}) admits the integral
\begin{equation}\label{character-fact}
\int_{\U(n)}\mu(\dv U) \phi(X_1UX_2U^\dagger,\mathbf{s})=\phi(X_1,\mathbf{s})\phi(X_2,\mathbf{s})
\end{equation}
for all $X_1,X_2\in\U(n)$, which has its origin in Schur's lemma.
This, in combination with the spherical transform
\begin{equation}\label{4.5}
\s f_X(\mathbf s)=\E[\phi(X,\mathbf s)]=\int_{\U(n)} \mu(\dv X) f_X(X)\phi(X,\mathbf s)
\end{equation}
of a random matrix $X\in\U(N)$, yields the desired convolution theorem
\begin{equation}\label{4.6}
\begin{split}
\s f_{X_1X_2}=\E_{X_1,X_2}[\phi(X_1X_2,\mathbf s)]=&\int_{\U(n)}\mu(\dv U)\E_{X_1,X_2}[\phi(UX_1U^\dagger X_2,\mathbf s)]\\
=&\E_{X_1,X_2}[\phi(X_1,\mathbf{s})\phi(X_2,\mathbf{s})]=\s f_{X_1}\cdot\s f_{X_2}
\end{split}
\end{equation}
where $X_1,X_2\in\U(n)$ are independent and one of them is invariant under the adjoint action of $\U(n)$.  In the second equality of~\eqref{4.6}, we assumed that $X_1$ satisfies the unitary invariance $f_{X_1}(X_1)=f_{X_1}(UX_1U^\dagger)$.

Another property of the characters is that they are unitarily invariant in their first entry, $\phi(X,\mathbf{s})=\phi(UXU^\dagger,\mathbf{s})$ for any $X,U\in\U(N)$. This 
in turn leads to the explicit Schur polynomial form \cite{Macdonald},
\begin{equation}\label{4.4}
\phi(X,\mathbf{s})=\phi(e^{i\bm \theta},\mathbf s)=\prod_{j=0}^{n-1}j!\,\frac{\det[e^{i\theta_js_k}]_{j,k=1}^n}{\Delta(e^{i\bm \theta})\Delta(\mathbf s)},
\end{equation}
where $e^{i\bm \theta}$ are the eigenvalues of $X\in\U(n)$. With the aid of this formula and the decomposition of the Haar measure~\cite[\S 4.2]{DF17}
\begin{equation}\label{4.2}
(\dv X)=\frac{1}{(2\pi)^nn!}|\Delta(e^{i\bm\theta})|^2(\dv \bm \theta)\mu(\dv U)
\end{equation}
for the eigenvalue decomposition $X=Ue^{i\bm\theta}U^\dagger$, the spherical transform of a unitarily invariant random matrix $X\in\U(n)$ reads
\begin{equation}\label{4.5.a}
\s f_X(\mathbf s)=\frac{\prod_{j=0}^{n-1}j!}{n!}\int_{(-\pi,\pi]^n} \frac{(\dv \bm \theta)}{(2\pi)^n}|\Delta(e^{i\bm\theta})|^2 f_X(e^{i\bm\theta})\frac{\det[e^{i\theta_js_k}]_{j,k=1}^n}{\Delta(e^{i\bm \theta})\Delta(\mathbf s)}.
\end{equation}

In this form three things can be read off. First, the normalisation reads $\s f_X(\mathbf s_1)=1$ with $\mathbf s_1=\diag(0,\ldots,n-1)$. Second, comparison of~\eqref{3.4} and~\eqref{3.1} tells us that there is another representation of the spherical transform as
	\begin{equation}\label{4.1}
	\s f_X(\mathbf s):=\int_{\U(n)}f_X(X)|X|^{\mathbf s}\,\mu(\dv X)
	\end{equation}
with $|X|^{\mathbf s}$ introduced in~\eqref{Selberg}. This is, however, only true when $s_j-s_{j+1}\geq 1$, because the integral may run through poles where the determinant of a subblock of $X$ vanishes. Thus, Eq.~\eqref{4.1} has to be employed carefully.
The third consequence of~\eqref{4.5.a} is an explicit inversion formula. Indeed such an inversion of the spherical transform can be given implicitly as a corollary of the theory of inverse spherical transforms, see~\cite{He84}. The following proposition is essentially given by the Peter-Weyl theorem and the Weyl character formula when considering spherical transforms for the Gelfand pair $(\U(n)\times \U(n), \diag(\U(n)\times \U(n)))$. Here we are only providing an explicit form for this particular case as we need it later.

\begin{proposition}\label{prop7}
	Let $f_X\in L^1(\U(n))$ and $\Delta(\mathbf s)^2\s f_X\in l^1$. Then, the inverse spherical transform is given by
	\begin{equation}\label{4.7}
	\s^{-1}(\s f_X)(e^{i\bm \theta})=\frac{1}{n!\prod_{j=0}^{n-1}(j!)^2}\sum_{\mathbf s\in \mathbb{Z}^n} \Delta(\mathbf s)^2 \s f_X(\mathbf s) \cdot\phi(e^{-i\bm \theta},\mathbf s),
	\end{equation}
	which evaluates to $f_X(e^{i\bm \theta})$ for almost all $\bm \theta\in(-\pi,\pi]^n$.
\end{proposition}

Let us point out that one can drop the condition $\Delta(\mathbf s)^2\s f_X\in l^1$ but one must pay the price of a regularising function in~\eqref{4.7} as is known for the Fourier transform or spherical transform on ${\rm Herm}_+(n)$; compared to~\cite[Lemma 2.10]{KK16} where a different regularisation with compact support is instead used.

\begin{proof}
It suffices to compute $\s^{-1}(\s f_X)(e^{i\bm\theta})$, explicitly. To apply Fubini's theorem, we introduce a Gaussian regularisation $e^{-\varepsilon^2\tr\,(\mathbf{s}-(n-1)\eins_n/2)^2}$ in the limit $\varepsilon\to 0$, so rendering everything absolutely convergent. Substituting~\eqref{4.4} and~\eqref{4.5.a} into~\eqref{4.7}, we can swap the sum with the integral which gives
	\begin{equation}\label{4.9}
	\begin{split}
	\s^{-1}(\s f_X)(e^{i\bm \theta})=&\frac{1}{(n!)^2}\lim_{\varepsilon\rightarrow 0}\int_{(-\pi,\pi]^n}\frac{(\dv\bm\varphi)}{(2\pi)^n} \frac{|\Delta(e^{i\bm \varphi})|^2}{\Delta(e^{i\bm \varphi})\Delta(e^{-i\bm \theta})} f_X(e^{i\bm\varphi})\\
	&\hspace*{-1.7cm}\times
	\left( \sum_{\mathbf s\in \mathbb{Z}^n} \left(\prod_{j=1}^ne^{-\varepsilon^2(s_j-(n-1)/2)^2}\right)\det[e^{i\varphi_js_k}]_{j,k=1}^n\det[e^{-i\theta_js_k}]_{j,k=1}^n\right).
	\end{split}
	\end{equation}
	 To evaluate the sum, we push the factor $e^{-\varepsilon^2(s_j-(n-1)/2)^2}$ in one of the determinants and, then, apply Andr\'eief's formula~\cite{An86} to obtain
	\begin{equation}\label{4.10}
	\begin{split}
	&\sum_{\mathbf s\in\Z^n}\left(\prod_{j=1}^ne^{-\varepsilon^2(s_j-(n-1)/2)^2}\right)\det[e^{i\varphi_js_k}]_{j,k=1}^n\det[e^{-i\theta_js_k}]_{j,k=1}^n
	\\
	=&n!\det\left[\sum_{s\in \mathbb{Z}} e^{-\varepsilon^2(s-(n-1)/2)^2}{e^{-i(\varphi_j-\theta_k)s}} \right]_{j,k=1}^n\\
	=&(2\pi)^nn!\det\left[g_\varepsilon(\varphi_j-\theta_k)\right]_{j,k=1}^n e^{-i(n-1)\tr(\bm\varphi-\bm\theta)/2},
	\end{split}
	\end{equation}
	where
\begin{equation}\label{heat-unitary}
g_{\varepsilon}(x):=\frac{1}{2\pi}\sum_{s\in \mathbb{Z}} e^{-\varepsilon^2(s-\frac{n-1}2)^2}{e^{-ix(s-\frac{n-1}2)}}=\sum_{m\in\Z}\frac{(-1)^{(n-1)m}}{2\sqrt{\pi}\varepsilon}e^{-\frac{(x+2m\pi)^2}{4\varepsilon^2}}.
\end{equation}
	The second equality in~\eqref{heat-unitary} follows from the Poisson summation formula. The function $g_\varepsilon(x)$ is a Jacobi theta function and the solution to a heat kernel on the finite interval $(-\pi,\pi]$, as shown in~\cite[Proposition 1.1]{LW2016}.
	
	The determinant can be replaced by the product of the diagonal entries only, due to the symmetry of the other determinants and $f_X$ under the permutations of $\bm\varphi$. This yields an additional factor $n!$, and we end up with
	\begin{equation}\label{4.11}
	\s^{-1}(\s f_X)(e^{i\bm \theta})=\lim_{\varepsilon\rightarrow 0}\int_{(-\pi,\pi]^n}\hspace*{-0.3cm}(\dv\bm\varphi)\, \prod_{1\leq j<l\leq n}\frac{\sin[(\varphi_l-\varphi_j)/2]}{\sin[(\theta_l-\theta_j)/2]} f_X(e^{i\bm\varphi})\prod_{j=1}^ng_\varepsilon(\varphi_j-\theta_k).
	\end{equation}
	Here, we have combined $e^{-i(n-1)\tr\bm\varphi/2}\Delta(e^{\bm\varphi})=\prod_{1\leq j<l\leq n}2i \sin[(\varphi_l-\varphi_j)/2$ and similar for $\bm \theta$.
	The last term is a product of the heat kernel on $(-\pi,\pi]$ that weakly asymptotes to $\prod_{j=1}^n\delta_{n}(\varphi_j-\vartheta_j)$ with the Dirac delta functions
	\begin{equation}\label{Dirac}
	\delta_n(\varphi)=\sum_{j=-\infty}^\infty[\delta(\varphi-4\pi j)-(-1)^n\delta(\varphi-4\pi j +2\pi)].
	\end{equation}
	The minus sign is important since it guarantees that the whole integrand is still $2\pi$ periodic which does not seem valid at first sight. Indeed for even $n$ the product $\prod_{1\leq j<l\leq n}2i \sin[(\varphi_l-\varphi_j)/2$ is only $4\pi$ periodic in each single angle which is compensated by the product of $g_\varepsilon(\varphi_j-\theta_k)$.
	 We conclude that $\s^{-1}(\s f_X)(e^{i\bm \theta})$is equal to $f_X(e^{i\bm \theta})$ almost everywhere.
\end{proof}

\subsection{Randomised Multiplicative Horn problem on $\U(n)$ and specialization to rank-$1$}\label{sec:Horn.unitary}

Analogous to Lemmas~\ref{l1} and~\ref{l4}, we have the following convolution theorem with fixed matrices.

\begin{lemma}[Multiplicative Convolution with Fixed Matrices on $\U(n)$]\label{lm14}
	Let $X$ be a random matrix in $\U(n)$, $\mathbf x_0$ be a diagonal matrix with complex phases, and $U$ be a Haar distributed $\U(n)$ matrix. Then, the spherical function of the product matrix $XU\mathbf x_0U^\dagger$ is
	\begin{equation}
	\s f_{XU\mathbf x_0U^\dagger}(\mathbf s)=\s f_X(\mathbf s)\cdot\phi(\mathbf x_0,\mathbf s).
	\end{equation}
\end{lemma}

\begin{proof}
This lemma is a direct consequence of~\eqref{character-fact} when averaging over $U\in\U(n)$ with respect to the Haar measure.
\end{proof}

As for the previous two cases, we need an auxiliary random matrix $H_\varepsilon$ to guarantee that all integrals and sums are absolutely convergent. The distribution of this random matrix unitarily invariant with the joint probability density of its eigenvalues $e^{i\mathbf x}=\diag(e^{ix_1},\ldots,e^{ix_n})$ with respect to the measure~\eqref{4.2} is
\begin{equation}\label{4.13a}
\begin{split}
f_{H_\varepsilon}(e^{i\mathbf x})=&\frac{(2\pi)^n}{\mathcal N_\varepsilon} \frac{e^{-i\frac{n-1}{2}\sum_{j=1}^nx_j}}{\Delta(e^{-i\mathbf x})}\det\left[
\left(i\frac{\partial}{\partial x_k}\right)^{j-1} g_{\varepsilon}(x_k)
\right]_{j,k=1}^n\\
=&\frac{(2\pi)^n}{\mathcal N_\varepsilon}  \frac{1}{\prod_{1\leq k<l\leq n}2\sin\left(\frac{x_l-x_k}{2}\right)}\det\left[
\left(-\frac{\partial}{\partial x_k}\right)^{j-1} g_{\varepsilon}(x_k)
\right]_{j,k=1}^n,
\end{split}
\end{equation}
with  $g_\varepsilon$ of Eq.~\eqref{heat-unitary}.  This has been proven~\cite[Proposition 1.1]{LW2016} for the initial condition of the identity matrix $\eins_n$, in particular we need to apply l'H\^opital's rule in~\cite{LW2016}.

When choosing two fixed diagonal matrices of phases $e^{i\mathbf a}=(e^{ia_1},\cdots,e^{ia_n})$ and $e^{i\mathbf b}=(e^{ib_1},\cdots,e^{ib_n})$, and two Haar distributed unitary matrices $U, V\in\U(n)$, we consider the product $C_\varepsilon=H_\varepsilon Ue^{i\mathbf a}U^\dagger Ve^{i\mathbf b}V^\dagger$ which has the spherical transform
\begin{equation}\label{4.12}
\s f_{C_\varepsilon}(\mathbf s)=\s f_{H_\varepsilon}(\mathbf s)\phi(e^{i\mathbf a},\mathbf s)\phi(e^{i\mathbf b},\mathbf s),
\end{equation}
due to~\eqref{lm14}.
Therefore, the joint eigenvalue probability density  of $C=C_0$ is given by
\begin{equation}\label{4.13}
p(e^{i\mathbf c})=\frac{|\Delta(e^{i\mathbf c})|^2}{(2\pi)^n\prod_{j=0}^{n}(j!)^2}\lim_{\varepsilon\rightarrow 0}\sum_{\mathbf s\in\Z^n}\Delta(\mathbf s)^2\,\s f_{H_\varepsilon}(\mathbf s)\,\phi(e^{i\mathbf a},\mathbf s)\phi(e^{i\mathbf b},\mathbf s)\phi(e^{-i\mathbf c},\mathbf s),
\end{equation}
which is the analogue of~\eqref{2.11} and~\eqref{3.112}. Although we have now sums instead of integrals, we run into the same problem as before when considering the most general setting, because of the three Vandermonde determinants $\Delta(\mathbf{s})$ that are hidden in the denominator of the spherical functions that stand against only two in the numerator. By cause of this, we recede to the simplest non-trivial realisation of Horn's problem on $\U(n)$, when $\mathbf{b}$ is rank-$1$.

In the rank-$1$ case, i.e. $e^{i\mathbf b}=(e^{ib},1,\cdots,1)$, the character $\phi(e^{i\mathbf b},\mathbf s)$ reduces to
\begin{equation}\label{4.14}
	\phi(e^{i\mathbf b},\mathbf s)=\frac{(n-1)!}{(1-e^{ib})^{n-1}}\sum_{p=1}^n\frac{e^{ibs_p}}{\prod_{\substack{l=1\\l\ne p}}^n(s_l-s_p)}.
\end{equation}
Hence, the joint density~\eqref{4.13} becomes
\begin{equation}\label{4.15}
\begin{split}
	p(e^{i\mathbf c})=&\frac{(n-1)!}{(2\pi)^n(n!)^2(1-e^{ib})^{n-1}}\frac{\Delta(e^{i\mathbf c})}{\Delta(e^{i\mathbf a})}
	\lim_{\varepsilon\rightarrow 0}\sum_{\mathbf s\in\Z^n}\left(\prod_{j=1}^ne^{-\varepsilon^2(s_j-\frac{n-1}{2})^2}\right)
	\\
	&\times
	\det[e^{-ic_js_k}]_{j,k=1}^n
	\det[e^{ia_js_k}]_{j,k=1}^n
	\sum_{p=1}^n\frac{e^{ibs_p}}{\prod_{\substack{l=1\\l\ne p}}^n(s_l-s_p)},
\end{split}
\end{equation}
where we drop the term  $\prod_{j=}^{n}e^{-\varepsilon^2(j-1-\frac{n-1}{2})^2}$  because it becomes unity in the limit $\varepsilon\to0$. It remains to carry out the series and the limit which we do in two step as for the additive Horn problem. 

\subsubsection*{Step 1: Application of Andr\'eief's identity}

The sum over $p$ can be resolved to the $s_1$-term due to permutation symmetry of the remaining terms in the sum and that the apparent poles at $s_l=s_p$ cancel with the zeros of the two determinants which all come with a multiplicity of at least $2$. Thus, we obtain a factor $n!$  and can exclude the value $s_1$ in the sums of $s_2,\ldots,s_n$. Both modifications allow us to exploit the generalised Andr\'eief formula~\cite[Eqn.~(C.4)]{KG10} for sums which leads to
\begin{equation}\label{4.20}
\begin{split}
p(e^{i\mathbf c})=&-\frac{i^{n-1}e^{i\frac{n-1}{2}(b+\tr(\mathbf{a}-\mathbf{c}))}}{n(e^{ib}-1)^{n-1}}\frac{\Delta(e^{i\mathbf c})}{\Delta(e^{i\mathbf a})}\lim_{\varepsilon\rightarrow 0}\sum_{s_1\in\Z}\,e^{-\varepsilon^2(s_1-\frac{n-1}{2})^2+ib(s_1-\frac{n-1}{2})}
\\
&\hspace*{-1cm}\times\det\begin{bmatrix}
0&\displaystyle\left[e^{-ic_k(s_1-\frac{n-1}{2})}\right]_{k=1,\cdots,n}
\\\\
\displaystyle\left[e^{ia_j(s_1-\frac{n-1}{2})}\right]_{j=1,\cdots,n}&\displaystyle\left[\frac{1}{2\pi i}\sum_{s\ne s_1}\frac{e^{-\varepsilon^2(s-\frac{n-1}{2})^2+i(a_j-c_k)(s-\frac{n-1}{2})}}{s_1-s}\right]_{j,k=1}^n
\end{bmatrix}.
\end{split}
\end{equation}
For the sum in the determinant, we can apply Parseval's theorem because we sum over a product of the two $l^2$-functions $e^{-\varepsilon^2(s-\frac{n-1}{2})^2}$ and $e^{i(a_j-c_k)(s-\frac{n-1}{2})}/(s_1-s)$. Then we get
\begin{equation}\label{pr1}
\begin{split}
&\frac{1}{2\pi i}\sum_{s\ne s_1}\frac{e^{-\varepsilon^2(s-\frac{n-1}{2})^2+i(a_j-c_k)(s-\frac{n-1}{2})}}{s_1-s}
\\
=&\int_{(-\pi,\pi]}\left(\sum_{s\in\Z} e^{-\varepsilon^2(s-\frac{n-1}{2})^2}e^{-ix(s-\frac{n-1}{2})}\right)\,\left(\sum_{s\ne s_1}\frac{e^{i(a_j-c_k)(s-\frac{n-1}{2})}}{s_1-s}e^{ix(s-\frac{n-1}{2})}\right)\frac{\dv x}{2\pi i}\\
=&\int_{(-\pi,\pi]}\left(\sum_{m\in\Z}\frac{(-1)^{(n-1)m}}{2\sqrt{\pi}\varepsilon}e^{-\frac{(x+2m\pi)^2}{4\varepsilon^2}}\right)e^{i(a_j-c_k+x)(s_1-\frac{n-1}{2})}\\
&\times({\rm mod}_{2\pi}(a_j-c_k+x)-\pi)\frac{\dv x}{2\pi}\\
=&\sum_{l\in\Z}\int_{c_k-a_j-2\pi l}^{c_k-a_j-2\pi  (l-1)} \frac{e^{-x^2/4\varepsilon^2+i(a_j-c_k+ x)(s_1-\frac{n-1}{2})}}{2\sqrt{\pi}\varepsilon}(a_j-c_k+x+(2l-1)\pi)\frac{\dv x}{2\pi},
\end{split}
\end{equation}
where we have, in the second equality, employed~\cite[3rd eqn. on Pg. 21]{Bo73}
\begin{equation}
\sum_{y\ne 0}\frac{e^{-i(a_j-c_k+x)y}}{y}={\rm ln}\left(\frac{1-e^{i(a_j-c_k+x)}}{1-e^{-i(a_j-c_k+x)}}\right)=i({\rm mod}_{2\pi}(a_j-c_k+x)-\pi)
\end{equation}
for all $(a_j-c_k-x)\notin 2\pi\Z$ and identified $g_\varepsilon$ so that we could use its second representation in~\eqref{heat-unitary}. In the third one, we have first evaluated the sum which became an integral over the whole real line and, with the ${\rm mod}_{2\pi}$ function then splitting the real line in another set of intervals.

Since $c_k-a_j$ always comes in the combination with $2\pi l$ apart from the phase factor $e^{-i(a_j-c_k)(s_1-\frac{n-1}{2})}$, we can shift this difference to $\Delta_{jk}={\rm mod}_{2\pi}(a_j-c_k)\in(0,2\pi)$ at the expense of a global sign. Then, we sum and integrate over a term with the prefactor ${\rm mod}_{2\pi}(a_j-c_k)+\pi+2i\varepsilon^2(s_1-(n-1)/2)$ of the Gaussian  yielding
\begin{equation}\label{pr2}
\begin{split}
&\sum_{l\in\Z}\int_{{\rm mod}_{2\pi}(a_j-c_k)-2\pi l}^{{\rm mod}_{2\pi}(a_j-c_k)-2\pi  (l-1)} \frac{e^{-x^2/4\varepsilon^2+i(a_j-c_k+ x)(s_1-\frac{n-1}{2})}}{2\sqrt{\pi}\varepsilon}\\
&\times\left({\rm mod}_{2\pi}(a_j-c_k)+\pi+2i\varepsilon^2\left(s_1-\frac{n-1}{2}\right)\right)\frac{\dv x}{2\pi}\\
=&\frac{1}{2\pi}\left({\rm mod}_{2\pi}(a_j-c_k)+\pi+2i\varepsilon^2\left(s_1-\frac{n-1}{2}\right)\right)e^{-\varepsilon^2(s_1-\frac{n-1}{2})^2+i(a_j-c_k)(s_1-\frac{n-1}{2})}.
\end{split}
\end{equation}
The remaining integral will be denoted by
\begin{equation}\label{pr1}
\begin{split}
f(\Delta_{jk};s_1)=&\sum_{l\in\Z}\int_{\Delta_{jk}-2\pi l}^{\Delta_{jk}-2\pi  (l-1)} \frac{e^{-(x-2i\varepsilon^2(s_1-\frac{n-1}{2}))^2/4\varepsilon^2}}{2\sqrt{\pi}\varepsilon}\\
&\times\left(x+2\pi (l-1)-2i\varepsilon^2\left(s_1-\frac{n-1}{2}\right)\right)\frac{\dv x}{2\pi}\\
=&\sum_{l\in\Z}\int_{\Delta_{jk}-2\pi}^{\Delta_{jk}} \hspace*{-0.3cm}\frac{e^{-(x-2\pi l-2i\varepsilon^2(s_1-\frac{n-1}{2}))^2/4\varepsilon^2}}{2\sqrt{\pi}\varepsilon}\left(x-2i\varepsilon^2\left(s_1-\frac{n-1}{2}\right)\right)\frac{\dv x}{2\pi}.
\end{split}
\end{equation}
In this way the density~\eqref{4.20} simplifies to
\begin{equation}\label{4.20a}
\begin{split}
p(e^{i\mathbf c})=&-\frac{i^{n-1}}{n(e^{ib}-1)^{n-1}}\frac{\Delta(e^{i\mathbf c})}{\Delta(e^{i\mathbf a})}\lim_{\varepsilon\rightarrow 0}\sum_{s_1\in\Z}\,
e^{-n\varepsilon^2(s_1+\frac{n-1}{2})^2+is_1(b+\sum_{j=1}^n(a_j-c_j))}
\\
&\hspace*{-1cm}\times\det\begin{bmatrix}
0&\displaystyle\left[1\right]_{1\times n}
\\\\
\displaystyle\left[1\right]_{n\times1} & \displaystyle\left[\frac{{\rm mod}_{2\pi}(a_j-c_k)}{2\pi}+f(\Delta_{jk};s_1)\right]_{j,k=1}^n
\end{bmatrix}.
\end{split}
\end{equation}
The term $[\pi+2i\varepsilon^2\left(s_1-(n-1)/2\right)]/(2\pi)$ has been subtracted either by the first row or column since it is independent of both indices.
What remains is to bound $f(\Delta_{jk};s_1)$ for suitable small $\varepsilon$ and then take the limit $\varepsilon\to0$.

\subsubsection*{Step 2: Limit $\varepsilon\to0$}

Since $\Delta_{jk}\in(0,2\pi)$, we can bound each summand in $f(\Delta_{jk};s_1)$. For $l\geq 1$ we have, up to two positive constants $\gamma_1$ and $\gamma_2$,
\begin{equation}
\begin{split}
&\left|\int_{\Delta_{jk}-2\pi}^{\Delta_{jk}} \frac{e^{-(x-2\pi l-2i\varepsilon^2(s_1-\frac{n-1}{2}))^2/4\varepsilon^2}}{2\sqrt{\pi}\varepsilon}\left(x-2i\varepsilon^2\left(s_1-\frac{n-1}{2}\right)\right)\frac{\dv x}{2\pi}\right|\\
\leq& \frac{\gamma_1+\gamma_2\varepsilon^2s_1}{\varepsilon}e^{\varepsilon^2(s_1-\frac{n-1}{2})^2-(\Delta_{jk}-2\pi l)^2/4\varepsilon^2},
\end{split}
\end{equation}
and for $l\leq-1$ the Gaussian term in $l$ has to be replaced by $e^{-(\Delta_{jk}-2\pi (l+1))^2/4\varepsilon^2}$.
Then, the partial sums of $l=1,2\ldots$ and $l=-1,-2,\ldots$ can be estimated by their corresponding integrals, i.e.
\begin{equation}
\begin{split}
\sum_{l=1}^\infty e^{-(\Delta_{jk}-2\pi l)^2/4\varepsilon^2}\leq & e^{-(\Delta_{jk}-2\pi)^2/4\varepsilon^2}+ \int_1^\infty dl e^{-(\Delta_{jk}-2\pi l)^2/4\varepsilon^2}\\
=&e^{-(\Delta_{jk}-2\pi)^2/4\varepsilon^2}+ \frac{\varepsilon}{2\sqrt{\pi}}{\rm erfc}\left[\left(1-\frac{\Delta_{jk}}{2\pi}\right)\frac{\pi}{\varepsilon}\right]\\
=&\mathcal{O}\left(e^{-(\Delta_{jk}-2\pi)^2/4\varepsilon^2}\right)
\end{split}
\end{equation}
and similarly for other sum,
\begin{equation}
\begin{split}
\sum_{l=-\infty}^{-1}e^{-(\Delta_{jk}-2\pi (l+1))^2/4\varepsilon^2}=&\mathcal{O}\left( e^{-\Delta_{jk}^2/4\varepsilon^2}\right).
\end{split}
\end{equation}
These estimates need the fact that $\Delta_{jk}$ stays away from $0$ and $2\pi$ for all combinations of $j,k=1,\ldots,n$. When we define $\alpha_{jk}=\min\{\Delta_{j,k},2\pi-\Delta_{j,k}\}$, the contribution of both sums can be written as
\begin{equation}\label{proof.b}
\begin{split}
&\left|\sum_{l\neq0}\int_{\Delta_{jk}-2\pi}^{\Delta_{jk}} \hspace*{-0.3cm}\frac{e^{-(x-2\pi l-2i\varepsilon^2(s_1-\frac{n-1}{2}))^2/4\varepsilon^2}}{2\sqrt{\pi}\varepsilon}\left(x-2i\varepsilon^2\left(s_1-\frac{n-1}{2}\right)\right)\frac{\dv x}{2\pi}\right|\\
=&\mathcal{O}\left(\frac{\gamma_1+\gamma_2\varepsilon^2s_1}{\varepsilon}e^{\varepsilon^2(s_1-\frac{n-1}{2})^2-\alpha_{jk}^2/4\varepsilon^2}\right).
\end{split}
\end{equation}

What remains to be estimated is the term $l=0$, which, however, can be evaluated directly since it is a total derivative,
\begin{equation}\label{proof.c}
\begin{split}
&\left|\int_{\Delta_{jk}-2\pi}^{\Delta_{jk}} \frac{e^{-(x-2i\varepsilon^2(s_1-\frac{n-1}{2}))^2/4\varepsilon^2}}{2\sqrt{\pi}\varepsilon}\left(x-2i\varepsilon^2\left(s_1-\frac{n-1}{2}\right)\right)\frac{\dv x}{2\pi}\right|\\
\leq& \frac{\varepsilon}{2\pi^{3/2}}\left(e^{-\Delta_{jk}^2/4\varepsilon^2}+e^{-(\Delta_{jk}-2\pi)^2/4\varepsilon^2}\right)e^{\varepsilon^2(s_1-\frac{n-1}{2})^2}\\
=&\mathcal{O}\left(\varepsilon e^{\varepsilon^2(s_1-\frac{n-1}{2})^2-\alpha_{jk}^2/4\varepsilon^2}\right).
\end{split}
\end{equation}
We combine this with~\eqref{proof.b} to find the bound of $f(\Delta_{jk};s_1)$ which is
\begin{equation}\label{4.28}
f(\Delta_{jk};s_1)=\mathcal{O}\left(\frac{\gamma_1+\gamma_2\varepsilon^2s_1}{\varepsilon}e^{\varepsilon^2(s_1-\frac{n-1}{2})^2-\alpha_{jk}^2/4\varepsilon^2}\right).
\end{equation}

Substituting~\eqref{4.28} into the $s_1$-sum in~\eqref{4.20}, we can expand the determinant in the correction $f(\Delta_{jk};s_1)$. All terms that involve this function have the form and estimate
\begin{equation}
\begin{split}
&\left|\sum_{s_1\in\Z}\,e^{-n\varepsilon^2(s_1-\frac{n-1}{2})^2+ibs_1}h(\mathbf{a},\mathbf{c})\prod_{l=1}^Lf(\Delta_{j_lk_l};s_1)\right|\\
\leq&|h(\mathbf{a},\mathbf{c})|\sum_{s_1\in\Z}\,e^{-(n-L)\varepsilon^2(s_1-\frac{n-1}{2})^2}\mathcal{O}\left(\frac{(\gamma_1+\gamma_2\varepsilon^2s_1)^L}{\varepsilon^L}e^{-L\alpha_{jk}^2/4\varepsilon^2}\right)\overset{\varepsilon\to0}{\rightarrow}0,
\end{split}
\end{equation}
where $h(\mathbf{a},\mathbf{c})$ only depends on $\mathbf{a}$ and $\mathbf{c}$ but not on $s_1$ or $\varepsilon$ and $L=1,\ldots,n-1$. The limit follows from the facts that $\min_{j,k}\{\alpha_{jk}\}>0$ and that the series have a finite limit when $\varepsilon\to0$.

Therefore the joint eigenvalue density takes the form
\begin{equation}\label{1.3a}
\begin{split}
p(e^{i\mathbf c})=&-\frac{i^{n-1}}{n(e^{ib}-1)^{n-1}}\frac{\Delta(e^{i\mathbf c})}{\Delta(e^{i\mathbf a})}\, \delta\left(b+\sum_{j=1}^n(a_j-c_j)\right)\\
&\times\det\begin{bmatrix}
0&\displaystyle[1]_{1\times n}
\\
\displaystyle[1]_{n\times 1}&\displaystyle\left[\frac{\mathrm{mod}_{2\pi}(a_j-c_k)}{2\pi}\right]_{j,k=1}^n\end{bmatrix}.
\end{split}
\end{equation}
We obtain the claim~\eqref{1.3} by considering a particular order, requiring a factor $n!$ to compensate. Additionally, we subtract with the first row times $a_j/(2\pi)$ and add with the first column times $c_k/(2\pi)$, when choosing $a_j,c_k\in(-\pi,\pi]$. Indeed, we have
\begin{equation}
\frac{\mathrm{mod}_{2\pi}(a_j-c_k)-a_j+c_k}{2\pi}=\Theta(c_k-a_j)
\end{equation}
when choosing only one period of $2\pi$ for  the length of an interval where the angles are drawn from. This finishes the proof.

\section{Conclusion}\label{S5}

In the present work, we have applied spherical transforms to the rank-$1$ randomised Horn problems in three settings, namely the three standard realizations of the symmetric spaces of the Lie algebra $\mathfrak u(N)$:
the Lie algebra itself as the additive group on $\mathrm{Herm}(n)$, the multiplicative action on the negatively curved space $\mathrm{Herm}_+(n)=\exp[\mathrm{Herm}(n)]$ and the multiplicative group $\U(n)=\exp[i\mathrm{Herm}(n)]$. In so doing, two important advancements have been achieved.

Firstly, we computed the analogue of the transition probability density to the rank-$1$ Horn problem on ${\rm Herm}_+(n)$, which has not been done before. The result looks strikingly similar to the corresponding additive Horn problem on ${\rm Herm}(n)$. Moreover, our approach outlined how to tackle more general randomised Horn problems and where the difficulties arise which have to be overcome for their solution.

Secondly, we laid out an approach to consider nontrivial random matrix ensembles on $\U(n)$, where one can study the product of random matrices on a one-dimensional domain, analogous to multiplication on ${\rm Herm}_+(n)$. Here
$\U(n)$ has the advantage that it is a group and a compact manifold. In particular, there is a unique uniform measure which is the Haar measure. Hence, we can expect new and different effects for products on such matrices.

Another intriguing question is how the results change when considering the other three classical Lie algebras and their three different realizations: odd and even dimensional spaces of anti-symmetric real matrices, and space of even-dimensional Hermitian self-dual matrices. One can already guess that then the new simplest non-trivial Horn-problem is the rank-$2$ case due to existing Weyl reflections, as known in the additive case \cite{FILZ17}. Nonetheless, we think they are still analytical feasible because they share similar algebraic structures like determinantal point processes with ${\rm Herm}(n)$. Relevant studies can be found in~\cite{CRZ20}.

\section*{Acknowledgements}

The work of PJF and JZ was supported by the Australian Research Council (ARC) through the ARC Centre of Excellence for Mathematical and Statistical frontiers (ACEMS). PJF also acknowledges partial support from ARC grant DP170102028, and JZ acknowledges the support of a Melbourne postgraduate award, and an ACEMS top up scholarship. We are grateful for fruitful discussions with Jesper R.~Ipsen. We also thank Vadim Gorin for correspondence leading to the penultimate
paragraph of the Introduction. And we also thank the referees for valuable comments.

\section*{Conflict of interest statement}
On behalf of all authors, the corresponding author states that there is no conflict of interest.

\providecommand{\bysame}{\leavevmode\hbox to3em{\hrulefill}\thinspace}
\providecommand{\MR}{\mathbb{R}elax\ifhmode\unskip\space\fi MR }
\providecommand{\MRhref}[2]{%
  \href{http://www.ams.org/mathscinet-getitem?mr=#1}{#2}
}
\providecommand{\href}[2]{#2}

\end{document}